\documentclass[sigconf,nonacm]{acmart}
\usepackage[utf8]{inputenc}

 \setcopyright{none}

\acmYear{}
\acmDOI{}
\acmISBN{}
\acmConference[ ]{}

\usepackage{amsfonts}
\usepackage{amsmath}
\usepackage{xspace}
\usepackage{cleveref}
\usepackage{xcolor}
\usepackage{tabularx}

\usepackage{dutchcal} 

\usepackage{graphicx}

\usepackage{caption}
\usepackage{subcaption}
\usepackage{comment} 
\usepackage{mathtools}

\usepackage[mathscr]{euscript}

\newtheorem{counterexample}{Counterexample}
\newtheorem*{remark}{Remark}

\newcommand{\timestep}{\ensuremath{\tau}\xspace}

\newcommand{\carlen}{\ensuremath{L}\xspace}

\newcommand{\mosa}{\ensuremath{\succeq_{MoS}^{PMC}}\xspace}
\newcommand{\mosaSMC}{\ensuremath{\succeq_{MoS}^{LSS}}\xspace}
\newcommand{\mosdist}{\ensuremath{\succeq_{d}}\xspace}
\newcommand{\mosspeed}{\ensuremath{\succeq_{v}}\xspace}
\newcommand{\moswl}{\ensuremath{\succeq_{\wl}}\xspace}

\newcommand{\ourpath}{\ensuremath{\pi}\xspace} 

\newcommand{\pset}{\ensuremath{\Pi}\xspace} 

\newcommand{\outlabs}{\ensuremath{\mathcal{O}}\xspace}

\newcommand{\union}{~\ensuremath{{\large \cup}\kern-0.5em\raisebox{0.75ex}{\tiny o}}~\xspace}

\newcommand{\intersec}{\xspace\ensuremath{{\large \cap}\kern-0.5em\raisebox{0.5ex}{\tiny o}}~\xspace}

\newcommand{\compl}{\xspace\ensuremath{{\large \setminus}\kern-0.2em\raisebox{0.5ex}{\tiny o}}~}

\newcommand{\oldstate}{\ensuremath{\bar{s}_b}\xspace}

\newcommand{\newstate}{\ensuremath{\bar{s}_a}\xspace}

\newcommand{\yesdet}{\ensuremath{\mathcal{D}}\xspace}
\newcommand{\nodet}{\ensuremath{\mathscr{M}}\xspace}
\newcommand{\emitdist}{\ensuremath{d}\xspace}

\newcommand{\oursafeprop}{\ensuremath{\psi_{nocol}}\xspace}
\newcommand{\oursafeprob}{\ensuremath{Pr_{nocol}}\xspace}

\newcommand{\detprob}{\ensuremath{Pr_{det}}\xspace}

\newcommand{\tankcount}{\ensuremath{J}\xspace}

\newcommand{\tanksize}{\ensuremath{TS}\xspace}

\newcommand{\wl}{\ensuremath{w}\xspace}

\newcommand{\wlp}{\ensuremath{\hat{w}}\xspace}

\newcommand{\inflow}{\ensuremath{in}\xspace}

\newcommand{\outflow}{\ensuremath{out}\xspace}

\newcommand{\uct}{\ensuremath{UT}\xspace}

\newcommand{\lct}{\ensuremath{LT}\xspace}

\newcommand{\err}{\ensuremath{e}\xspace}

\newcommand{\errwin}{\ensuremath{EW}\xspace}

\newcommand{\fulltank}{\ensuremath{FU}\xspace}

\newcommand{\emptytank}{\ensuremath{EM}\xspace}

\newcommand{\tanksafeprop}{\ensuremath{\psi_{flow}}\xspace}

\newcommand{\timebound}{\ensuremath{T}\xspace}

\newcommand{\mdl}{\ensuremath{\mathsf{M}}\xspace}

\newcommand{\mcp}{\ensuremath{\mathsf{M}_{\text{cpl}}}\xspace}

\newcommand{\mper}{\ensuremath{\textsf{M}_{\text{per}}}\xspace}

\newcommand{\mba}{\ensuremath{\mathsf{M}_{\text{int}}}\xspace}

\newcommand{\mbatt}{\ensuremath{\mathsf{M}_{\text{tt}}}\xspace}

\newcommand{\mbattneg}{\ensuremath{\mathsf{M}_{\text{ttneg}}}\xspace}

\newcommand{\specialcell}[2][c]{%
  \begin{tabular}[#1]{@{}c@{}}#2\end{tabular}}

\crefname{counterexample}{Counterexample}{Counterexamples}

\newcommand\IR[1]{$\dag$\footnote{IR: #1}}

\begin{document}

\title{Monotonic Safety for \\Scalable and Data-Efficient Probabilistic Safety Analysis}

 \author{Matthew Cleaveland}
 \affiliation{%
   \institution{
 University of Pennsylvania}
   \city{Philadelphia}
   \country{USA}}
 \email{mcleav@seas.upenn.edu}

 \author{Ivan Ruchkin}
 \affiliation{%
   \institution{
 University of Pennsylvania}
   \city{Philadelphia}
   \country{USA}}
 \email{iruchkin@seas.upenn.edu}

 \author{Oleg Sokolsky}
 \affiliation{%
   \institution{
 University of Pennsylvania}
   \city{Philadelphia}
   \country{USA}}
 \email{sokolsky@seas.upenn.edu}

 \author{Insup Lee}
 \affiliation{%
   \institution{
 University of Pennsylvania}
   \city{Philadelphia}
   \country{USA}}
 \email{lee@seas.upenn.edu}

\ccsdesc[500]{Software and its engineering~Model checking}
\ccsdesc[300]{Theory of computation~Formal languages and automata theory}

\keywords{probabilistic model checking, statistical model checking, monotonic safety}

\begin{abstract}
    
    Autonomous systems with machine learning-based perception can exhibit unpredictable behaviors that are difficult to quantify, let alone verify. Such behaviors are convenient to capture in probabilistic models, but probabilistic model checking of such models is difficult to scale --- largely due to the non-determinism added to models as a prerequisite for provable conservatism. Statistical model checking (SMC) has been proposed to address the scalability issue. However it requires large amounts of data to account for the aforementioned non-determinism, which in turn limits its scalability. This work introduces a general technique for reduction of non-determinism based on assumptions of ``monotonic safety'', which define a partial order between system states in terms of their probabilities of being safe. We exploit these assumptions to remove non-determinism from controller/plant models to drastically speed up probabilistic model checking and statistical model checking while providing provably conservative estimates as long as the safety is indeed monotonic. Our experiments demonstrate model-checking speed-ups of an order of magnitude while maintaining acceptable accuracy and require much less data for accurate estimates when running SMC --- even when monotonic safety does not perfectly hold and provable conservatism is not achieved.
    

\end{abstract}

\maketitle

\section{Introduction}
\label{sec:intro}

Model checking is a well-established way to assure a safety-critical CPS, such as a self-driving car, by searching for violations of a safety property in the reachable states of a system model. Recent advances in perception based on machine learning (ML) have challenged model checking with unpredictable, difficult-to-model behaviors~\cite{faria_non-determinism_2017,burton_making_2017,yampolskiy_predicting_2018}. The uncertainties of ML-based perception and the environment where it is deployed call for probabilistic model checking (PMC)~\cite{SegalaThesis,kwiatkowska_assume-guarantee_2010,Feng11}, which computes a probability that a property holds (e.g., the chance of no collisions). However, unpredictable outputs and qualitative perception errors~\cite{Chhokra20}, such as false negatives, can introduce complex behaviors which lead to scalability issues when performing PMC for CPS with ML components.


A common way of addressing these scalability issues is with abstractions, which simplify system models while preserving their properties of interest. For example, instead of treating a system's state space as continuous, one can grid it up into intervals, thus shrinking the size of the model and reducing the run time of probabilistic model checking. However, many abstraction techniques rely on non-determinism to ensure conservatism (e.g., one state grid cell can transition to multiple next grid cells on a given input). The model checker must then account for these non-deterministic choices, which then leads to new scalability issues.

SMC offers an alternative to PMC. SMC samples traces of a probabilistic model and uses statistical methods to obtain guarantees about the model's behavior. In the case of probabilistic models with non-determinism, this is done by employing lightweight scheduler sampling (LSS)~\cite{legay2014LSS}. However, LSS needs to separately sample both the probabilistic and non-deterministic behaviors of the system in question, which makes it very data hungry in the presence of large amounts of non-determinism.

To improve the scalability of PMC and data efficiency of LSS, this paper introduces a novel way to reduce non-determinism in the intermediate abstractions of models. This reduction is based on \emph{assumptions of monotonic safety} (MoS) that capture our intuition about which states have a higher chance of safety than others. Encoded as partial orders, these assumptions draw on domain-specific knowledge to characterize broad safety principles in a particular system. For instance, in an obstacle-avoidance scenario, being further away from an obstacle is generally safer than being closer to it. MoS assumptions can be used to simplify models. These simplifications intend to maintain conservative safety estimates and, when MoS assumptions hold, are provably conservative. Specifically, we simplify models using MoS-driven \emph{trimming} of non-deterministic transitions leading to MoS-safer states. 
This reduction in non-determinism speeds up PMC and improves the data efficiency of LSS.

Ultimately, MoS assumptions are heuristic and may not strictly hold in all situations: in rare cases it is worth being a little closer to the obstacle so that the perception is more accurate. Typically, it is not immediately clear where MoS holds or how to efficiently compute that MoS holds. However, we show empirically that MoS assumptions do not need to strictly hold in every state to be useful, as long as they hold most of the time.


We instantiate our MoS assumptions and MoS-driven trimming on two simulated case studies: automated braking for an autonomous vehicle and flow control of water tanks. Our MoS-driven trimming produces abstractions which estimate the safety probability of the system and are compared with standard conservative abstractions. Experiments show that our abstractions improve the scalability of PMC by an order of magnitude and the data efficiency of LSS by at least an order of magnitude. Moreover, we observe that our abstractions remain empirically conservative, even when not provably so: they do not overestimate the chance of safety in practice, even though MoS assumptions only hold in most of the states. 

This paper makes three contributions: (i) a novel notion of MoS assumptions and their use for transition trimming, (ii) conservatism proofs for transition trimming under MoS assumptions, (iii) an application of MoS-trimmed abstractions to emergency braking systems and flow control in water tanks, resulting in  high-performing instances of abstractions and quantifications of the MoS assumptions.

The rest of the paper is organized as follows. We introduce two complementary motivating systems in \Cref{sec:system} and give the necessary formal background in \Cref{sec:background}. \Cref{sec:assn} introduces MoS assumptions. \Cref{sec:approach} explains how we exploit MoS assumptions to trim non-determinism from models. We describe the results of our two case studies in \Cref{sec:evaluation}. The paper ends with related work (\Cref{sec:related}) and a conclusion (\Cref{sec:discussion}).

\section{Motivating Systems}
\label{sec:system}

Throughout the paper, we consider two systems with complementary complexities. The first one, \emph{emergency braking}, has a stateless multi-output controller, a stateful binary-output perception, and a total-order MoS in each state dimension. The second one, \emph{water tanks}, has a stateful binary-output controller, a stateless multi-output perception, and a U-shaped MoS assumption. 

\subsection{First System: Emergency Braking} 

Consider an autonomous car approaching a stationary obstacle. The car is equipped with a controller which issues braking commands on detection of the obstacle and an ML-based perception system which detects obstacles. The safety goal of the car is to fully stop before hitting the obstacle. Following are the car dynamics, controller, and perception used in the paper.

A discrete kinematics model with time step \timestep represents the velocity ($v$) and position ($d$, same as the distance to the obstacle) of the car at time $t$:
\begin{equation} \label{eq:aebs-dyn}
    d[t] = d[t-\timestep] - \timestep \times v[t-\timestep], \quad
    v[t] = v[t-\timestep] - \timestep \times b[t-\timestep],
\end{equation}
where $b[t]$ is the braking command (a.k.a. the ``braking power'', BP) at time $t$.

We use an Advanced Emergency Braking System (AEBS)~\cite{aebs} which uses two metrics to determine the BP: time to collide ($TTC$) and warning index ($WI$). The $TTC$ is the amount of time until a collision if the current velocity is maintained. The $WI$ represents how safe the car would be in the hands of a human driver (positive is safe, negative is unsafe).
\begin{equation}
    TTC = \frac{d}{v},\quad WI = \frac{d-d_{br}}{ v \cdot T_{h}}, \quad d_{br} = v\cdot T_{s}+ \frac{u\cdot v^2}{2a_{max}}, 
\end{equation}
where $d_{br}$ is the braking-critical distance, $T_{s}$ is the system response delay (negligible), $T_{h}$ is the average driver reaction time (set to 2s), $u$ is the friction scaling coefficient (taken as 1), and $a_{max}$ is the maximum deceleration of the car.

Upon detecting the obstacle, the AEBS chooses one of three BPs: no braking (BP of 0), light braking ($B_1$), and maximum braking ($B_2=a_{max}$). If the obstacle is not detected, no braking occurs. The  BP value, $b$, is determined by $WI$ and $TTC$ crossing either none, one, or both of the fixed thresholds $C_1$ and $C_2$:
\begin{align*}
   WI >C_1 \land TTC > C_2 \implies & b = 0 \\
    (WI \leq C_1 \land  TTC > C_2) \lor (WI > C_1 \land  TTC \leq C_2) \implies & b = B_1  \\
    WI \leq C_1 \land  TTC \leq C_2 \implies & b = B_2
\end{align*}

To detect the obstacle, the car uses the deep neural network YoloNetv3~\cite{yolonet},
as it can run at high frequencies~\cite{why_yolo}. We say a low-level detection of the obstacle occurs when Yolo detects the obstacle. To reduce noise, we apply a majority vote filter to the low-level detections, so the AEBS receives a high-level detection with the distance to the obstacle if at least 2 of the past 3 Yolo outputs were low-level detections. We model Yolo probabilistically based on two observations. First, Yolo's chance of detection is higher when the car is closer to the obstacle. Second, consecutive low-level detections correlate with each other due to weather conditions and similar car positions. Our perception model conditions the detection chance on the distance from the obstacle and the recent detection history.

The safety property of interest, \oursafeprop, is the absence of a collision, specified in linear temporal logic 
(LTL)~\cite{pnueli_temporal_1977} as 
\begin{equation}
    \oursafeprop :=  \square \; \left ( d > \carlen \right ),
\label{eq:aebs_safety_prop}
\end{equation}
where \carlen is the minimum allowed distance to the obstacle and is taken to be 5m. Note that in this model the car eventually stops or collides, so for any initial condition there is an upper bound on the number of time steps.

\subsection{Second System: Water Tanks}
Consider a system consisting of \tankcount water tanks, each of size \tanksize, draining over time and a controller that maintains some water level in each tank. With $\wl_i[t]$ as the water level in the $i^{th}$ tank at time $t$, the discrete time dynamics for the water level in the tank at the next time step is given by:
\begin{equation}
    \wl_i[t+1] = \wl_i[t] - \outflow_i[t] +  \inflow_i[t],
\end{equation}
where $\inflow_i[t]$ and $\outflow_i[t]$ are the amounts of water entering (``inflow'') and leaving (``outflow'') respectively the $i^{th}$ tank at time $t$. The inflow is determined by the controller and the outflow is a constant determined by the environment.

Each tank is equipped with ML-based perception to report its current perceived water level, \wlp, which is a random function of the true current water level, \wl. In the simulated system (playing the role of the reality), \wlp is determined as a positively-skewed high-variance mixture-of-gaussians error added to \wl. The gaussians have higher variance for water levels in the middle of the tank and smaller variance near the edges of the tank. In addition, with constant probability the perception can output \wlp=0 or \wlp=\tanksize (to account for a large range of perception errors). 

We model perception probabilistically as a categorical distribution of perception error, \err. The possible error values are integers within some \errwin-sized error window, 
$$-\errwin,\dots, -1,  0, 1, \dots, \errwin,$$ 
plus two special values to capture unexpected outputs of ML-based perception (e.g., due to water reflections or deep network peculiarities): $\fulltank$ to indicate that a spurious reading of the tank being full occurred (then $\wlp=\tanksize$), and $\emptytank$ to indicate that a spurious reading of the tank being empty occurred (then $\wlp=0$). Hence, this perception model has $2\errwin + 3$ parameters.   

The controller has a single source of water to fill one tank at a time (or none at all) based on the perceived water levels. Then this tank receives a constant value $\inflow>0$ of water, whereas the other tanks receive 0 water. The controller only starts filling a tank that has water below the lower decision threshold \lct and stops filling after the water reaches the upper threshold \uct. After filling a tank or taking one step of not filling, if the controller perceives multiple tanks below \lct, it starts filling the one with the least perceived water (or, if equal, with the lowest ID 1\dots\tankcount). 

The safety property, \tanksafeprop, is that each tank must never run out of water (i.e., $\wl = 0$) or overflow (i.e., $\wl = \tanksize$).  This property can be expressed as the following time-bounded LTL formula:
\begin{equation}
\label{eq:tank_safety_prop}
     \tanksafeprop :=  \square_\timebound \; \left (0 < \wl_1 < \tanksize \land \dots \land 0 < \wl_\tankcount < \tanksize \right )
\end{equation}
where \timebound is a fixed number of time steps within which the tanks should be kept is the water level resulting in overflow of tanks.

\section{Background, System Models, and LSS}
\label{sec:background}


In the following \Cref{def:pa,def:parcomp,def:prob}, borrowed from Kwiatkowska et al. \cite{Kwiatkowska2013}, we use $Dist(S)$ to refer to the set of probability distributions over a set $S$, $\eta_s$ as the distribution with all its probability mass on $s \in S$, and $\mu_1 \times \mu_2$ to be the product distribution of $\mu_1$ and $\mu_2$.
\begin{definition} \label{def:pa}
A \emph{probabilistic automaton} 
(PA) is a tuple $\mdl=(S,\bar{s},\alpha,\delta,L)$, where $S$ is a finite set of states, $\bar{s}\in S$ is the initial state, $\alpha$ is an alphabet of action labels, $(S,\alpha,Dist(S)) \in \delta$ is a probabilistic transition relation, and $L: S \rightarrow 2^{AP}$ is a labeling function from states to sets of atomic propositions from the set AP. 
\end{definition} 
If $(s,a,\mu)\in \delta$ then the PA can make a transition in state $s$ with action label $a$ and move based on distribution $\mu$ to state $s'$ with probability $\mu(s')$, which is denoted by $s \xrightarrow{a} \mu$. If $(s,a,\eta_{s'})\in \delta$ then we say the PA can transition from state $s$ to state $s'$ via action $a$. A state $s$ is terminal if no elements of $\delta$ contain $s$. A path in $M$ is a finite/infinite sequence of transitions $\ourpath = s_0 \xrightarrow{a_0,\mu_{0}} s_1 \xrightarrow{a_1, \mu_{1}}\hdots$ with $s_0=\bar{s}$ and $\mu_{i}(s_{i+1})>0$.
A set of paths is denoted as $\pset$. We use $\mdl(s)$ to denote the PA $\mdl$ with initial state $s$.


Reasoning about PAs also requires the notion of \textit{schedulers}, which resolve the non-determinism during an execution of a PA. For our purposes, a scheduler $\sigma$ maps each state of the PA to an available action label in that state. We use $\pset_{\mdl}^{\sigma}$ for the set of all paths through $\mdl$ when controlled by scheduler $\sigma$ and $Sch_{\mdl}$ for the set of all schedulers for $\mdl$. Finally, given a scheduler $\sigma$, we define a probability space $Pr_{\mdl}^{\sigma}$ over the set of paths $\pset_{\mdl}^{\sigma}$ in the standard manner. 

Given PAs $\mdl_1$ and $\mdl_2$, we define parallel composition as:

\begin{definition}\label{def:parcomp}
    The \emph{parallel composition}    of PAs $\mdl_1=(S_1,\bar{s}_1,\alpha_1,\delta_1,L_1)$ and $\mdl_2=(S_2,\bar{s}_2,\alpha_2,\delta_2,L_2)$ is given by the PA $\mdl_1||\mdl_2=(S_1 \times S_2, (\bar{s}_1,\bar{s}_2),\alpha_1 \cup \alpha_2, \delta,L)$, where $L(s_1,s_2)=L_1(s_1)\cup L_2(s_2)$ and $\delta$ is such that $(s_1,s_2) \xrightarrow{a} \mu_1 \times \mu_2$ iff one of the following holds: (i) $s_1\xrightarrow{a}\mu_1,s_2\xrightarrow{a}\mu_2$ and $a \in \alpha_1 \cap \alpha_2$, (ii) $s_1 \xrightarrow{a}\mu_1, \mu_2=\eta_{s_2}$ and $a \in (\alpha_1 \setminus \alpha_2)$, (iii) $\mu_1=\eta_{s_1},s_2\xrightarrow{a}\mu_2$ and $a \in (\alpha_2 \setminus \alpha_1)$.
\end{definition}

In this paper, we are concerned with probabilities of safety properties, which we state using LTL formulas over state labels.

\begin{definition}\label{def:prob}
For LTL formula $
\psi$, PA $\mdl$, and scheduler $\sigma \in Sch_{\mdl}$, the \emph{probability of $\psi$ holding} is:
\begin{align*}
    Pr_{\mdl}^{\sigma}(\psi) \coloneqq Pr_{\mdl}^{\sigma}\{ \pi \in \pset_{\mdl}^{\sigma}~|~\pi \models \psi \}
\end{align*}
where $\pi \models \psi$ indicates that the path $\pi$ satisfies $\psi$ in the standard LTL semantics~\cite{pnueli_temporal_1977}. We specifically consider LTL safety properties, which are LTL specifications that can be falsified by a finite trace though a model. Both \oursafeprop and \tanksafeprop are LTL safety properties.
\end{definition}

Probabilistically verifying an LTL formula $\psi$ against $M$ requires checking that the probability of satisfying $\psi$ meets a probability bound for all schedulers. This involves computing the minimum or maximum probability of satisfying $\psi$ over all schedulers:
\begin{align*}
    Pr_{\mdl}^{min}(\psi) &\coloneqq \operatorname{inf}_{\sigma\in Sch_{\mdl}} Pr_{\mdl}^{\sigma}(\psi) \\ Pr_{\mdl}^{max}(\psi) &\coloneqq \operatorname{sup}_{\sigma\in Sch_{\mdl}} Pr_{\mdl}^{\sigma}(\psi)
\end{align*}

We call $\sigma$ a min scheduler of \mdl if $Pr_{\mdl}^{\sigma}(\psi) = Pr_{\mdl}^{min}(\psi)$. We use $Sch_{\mdl}^{min}$ to denote the set of min schedulers of \mdl.

\subsection{Probabilistic Automata for Motivating Systems}
\label{sec:modelBackground}

\begin{figure}
\centering
  \includegraphics[width=0.8\linewidth]{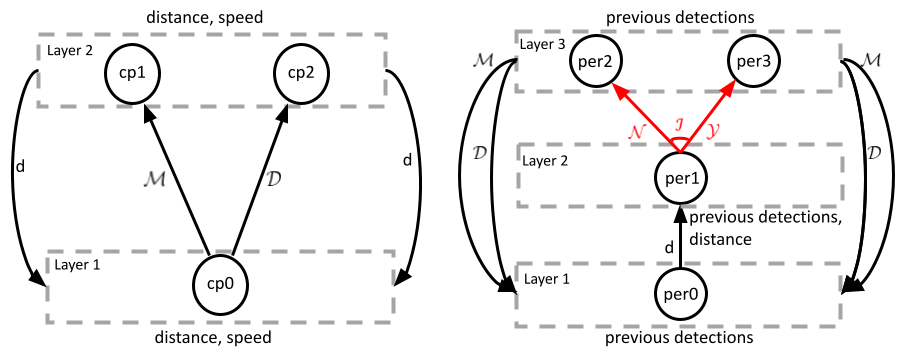}
  \caption{\small The models of controller/plant \mcp (left) and perception \mper (right) in the emergency braking system. The dashed grey boxes indicate layers of states, with one state for each pair of distances and speeds. The red transition is probabilistic. }
\label{fig:McpMper}
\vspace*{-3mm}
\end{figure}


We now describe how we construct the baseline PAs for the systems described in \Cref{sec:system}. We always start with two models: a non-probabilistic, discrete time, infinite state controller-plant model \mcp and a probabilistic perception model \mper. Our goal is to abstract them as PAs and compose them so that we can analyze their safety properties. As an illustration, we discuss our models for the emergency braking system, shown in \Cref{fig:McpMper}. \mper probabilistically generates either detections or non-detections of the obstacle and sends them to \mcp, which then selects a braking command in response to the detection/non-detection and updates its (continuous) distance and speed values accordingly. \mcp then sends the distance to the obstacle back to \mper, which it uses to alter its detection/non-detection probabilities. We now describe how we construct \mper as a PA before describing how to form a standard PA abstraction of \mcp. Note that \mcp is not a PA since it has an infinite number of states.


\mper has three actions --- \yesdet (detection), \nodet (non-detection), \emph{distance=d} (\emitdist), all of which are synchronized with \mcp. Colloquially, we say that \mper transmits actions \yesdet and \nodet to \mcp and \mcp transmits parameterized action \emitdist to \mper. \mper conditions its detection probabilities on both the distance \emitdist and its last $W$ low-level readings and uses a majority-vote filter over its last $N_{F}$ low-level readings to determine its output of either \yesdet or \nodet, as described in \Cref{sec:system}. The choice of low-level detection is the only probabilistic transition in \mper (and the whole system). \mper waits for \mcp to transmit its distance. Then it makes an internal detection and checks if at least half of the previous $N_{F}$ low-level readings were detections. If so, it transmits \yesdet. If not, it transmits \nodet. Either way, it then waits for \mcp to transmit the new distance \emitdist.

\begin{figure}
\centering
\begin{subfigure}{0.23\textwidth}
  \centering
  \includegraphics[width=1\linewidth]{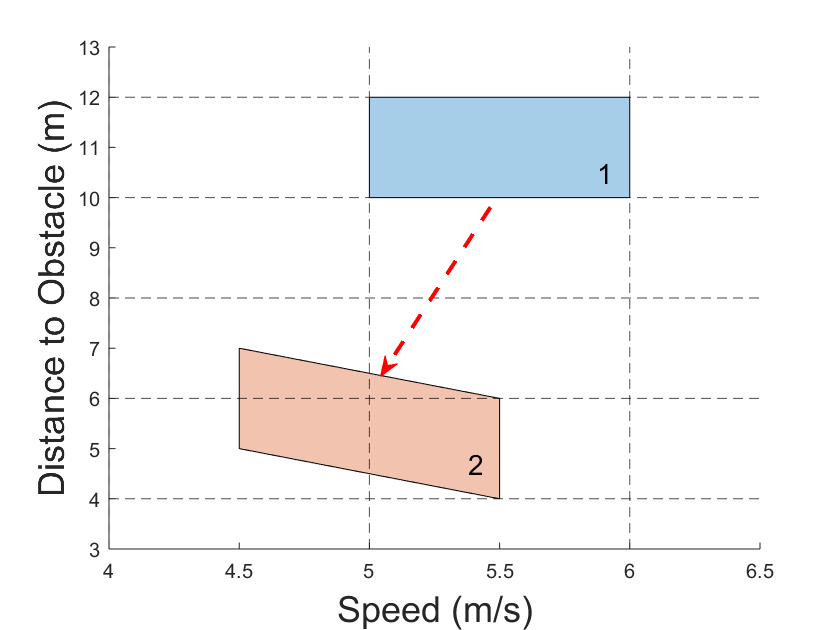}
  \caption{}
  \label{fig:mint_construction_1}
\end{subfigure}%
\begin{subfigure}{0.23\textwidth}
  \centering
  \includegraphics[width=1\linewidth]{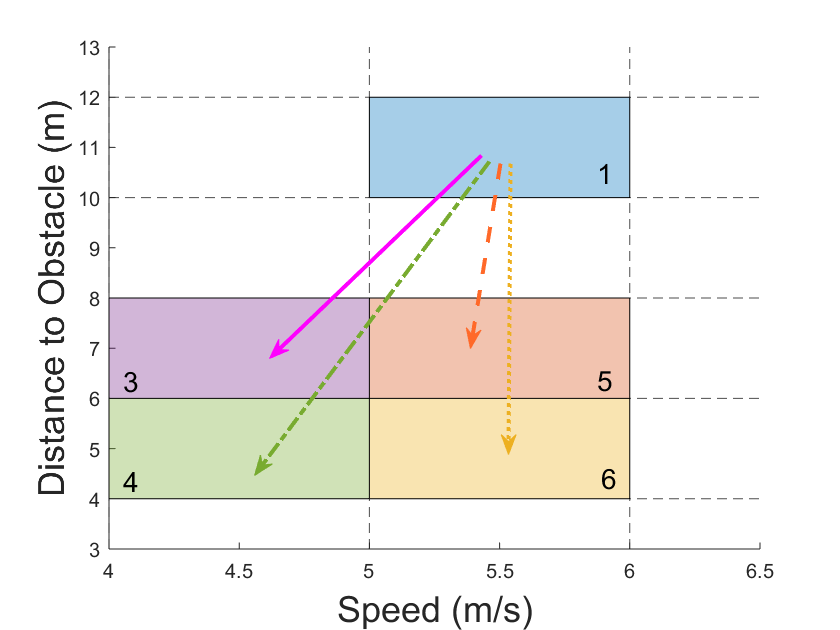}
  \caption{}
  \label{fig:mint_construction_2}
\end{subfigure}
\caption{\small In \Cref{fig:mint_construction_1} a set of distance and speed states (box 1) are grouped into a single entity and the set of possible next states (polygon 2) is computed. In \Cref{fig:mint_construction_2} the set of next states is used to compute the new transitions from box 1, which go to boxes 3-6.}
\label{fig:mint_construction}
\vspace*{-3mm}
\end{figure}

We now describe \mcp and then create PA abstraction \mba. Let $S \subset \mathbb{R}^n$ be the state space of \mcp, $\bar{s}$ be the initial state, $K=\{ \yesdet, \nodet\}$ be the set of \mper inputs, and $f: S \times K \rightarrow S$ be the discrete transition relation (i.e. if \mcp is in state $s\in S$ and receives high-level detection $k \in K$ from \mper, then at the next time point it will be in state $f(s,k)$).

The state space of \mba, denoted $S'$, is formed by dividing $S$ into a discrete, finite set of regions using equally sized intervals\footnote{The size of these intervals is a modeling hyperparameter. We explore how it affects the abstractions in \Cref{sec:evaluation}.} (see  the squares formed by the dashed lines in \Cref{fig:mint_construction_1}). So every $s_1' \in S'$ has a corresponding region $S_1 \subset S$. The initial state $\bar{s}'$ is the state in $S'$ which contains $\bar{s}$ in its set of corresponding states. The set of \mper inputs $K=\{ \yesdet, \nodet\}$ is unchanged. The transition relation $\delta'$ is formed as follows:
\begin{align}
    &(s_1',(k \times \alpha'),\eta_{s_2'}) \in \delta' \; \text{if} \\ 
    &\exists s_2' \in S', \; \exists s_1 \in S_1, s_2 \in S_2, \;  s.t. f(s_1,k)=s_2 \nonumber
\end{align}
In other words, \mba has a transition from $s_1'$ to $s_2'$ (in \mba) if at least one state in $S_1$ has a transition to a state in $S_2$ (in \mcp). For an illustration of this, see \Cref{fig:mint_construction}. Note that $\delta'$ is non-deterministic due to the perception values $K$ and the reachability analysis over the intervals of states in \mcp. The latter of two types of non-determinism is captured by the extra action labels $\alpha' \in A'$. Each instance of the reachability non-determinism gets its own unique action label from $A'$ (see the different arrows in \Cref{fig:mint_construction_2}) and the perception non-determinism will be resolved when the model gets composed with \mper. Whenever \mba makes a transition, it transmits the distance range corresponding to its next state $s_2'$ for \mper to synchronize on. The discrete states $s_i' \in S'$ inherit any atomic labels $l\in L$ from the states in $S_i$, leading to labelling function $L'$. Finally, we convert \mba into a PA using the tuple notation from definition \Cref{def:pa}: $\mba = \{ S',s_0',K\times A',\delta',L'\}$.

We now form our baseline model, which we also denote as \mba, by taking $\mper || \mba$. Note that the resulting model has non-determinism in its transition relation, which our MoS assumptions will aim to remove.

For the water tank system, the \mper and \mba models are constructed from equations in \Cref{sec:system} in a similar fashion: \mper encodes the perception probabilistically and \mba is a PA based interval abstraction of \mcp encoding the controller and plant non-probabilistically. \mba communicates the true water levels $(\wl_1 \dots \wl_\tankcount)$ to \mper and receives the perceived water levels $(\wlp_1 \dots \wlp_\tankcount)$. The execution terminates either when the safety property is violated or time \timebound is reached. 

\subsection{Statistical Model Checking for PAs}
\label{sec:LSS}

This section formalizes the process of running LSS on PA \mdl to approximate $Pr_{\mdl}^{min}(\psi)$. First, a set of schedulers are sampled uniformly from $Sch_{\mdl}$:
\begin{align*}
    \sigma_1, \hdots, \sigma_n \sim Unif(Sch_{\mdl})
\end{align*}

Then the probability of $\mdl$ under scheduler $\sigma_i$, which we denote as $p_i \coloneqq Pr_{\mdl}^{\sigma_i}(\psi)$, is computed via sampling traces from the fully probabilistic model of \mdl when scheduler $\sigma_i$ is fixed. Finally, the smallest probability is returned: $Pr_{\mdl}^{LSS,n}(\psi)=min(p_1,\hdots,p_n)$.


Since LSS relies on sampling schedulers, its output is a random variable. So consecutive runs of LSS can give different results. In addition, running LSS requires $n$ iterations of standard statistical analysis of \mdl with fixed schedulers. Reducing $n$ lowers the run times of LSS, but will also give less conservative results on average. The ideal choice of $n$ depends both on the computational resources at hand and the estimated size of the set of schedulers $|Sch_{\mdl}|$. LSS will never output a smaller value than PMC, so $Pr_{\mdl}^{LSS,n}(\psi) \geq Pr_{\mdl}^{min}(\psi)$. However, as $n \rightarrow \infty$, $P\left(Pr_{\mdl}^{LSS,n}(\psi) = Pr_{\mdl}^{min}(\psi)\right)\rightarrow 1$, which is to say that with enough samples LSS will eventually find the smallest satisfaction probability. In practice $|Sch_{\mdl}|$ will be extremely large, thus requiring very large $n$ to obtain reasonable safety estimates.



\begin{remark}
    For a complete description of LSS, see \cite{legay2014LSS,DArgenio2015SmartLSS}.
\end{remark}
\section{Assumptions of Monotonic Safety}\label{sec:assn}

This section starts by formalizing assumptions of \emph{monotonic safety} (MoS) for both PMC and LSS\footnote{The two techniques require different MoS assumptions because they handle non-determinism differently.} and stating the intuitive MoS assumptions for the motivating systems. Then we observe that MoS does not always hold in practice, exemplified in our motivating systems. 

\begin{definition}[Assumption of Monotonic Safety for PMC]\label{def:mosPMC}
An \emph{MoS assumption for PMC} is a partial order, \mosa, over the states $S$ with respect to a given safety property $\psi$ and a PA \mdl. It stipulates that if two states $s_1, s_2 \in S$ are ordered, $s_1 \mosa s_2$, then for every min scheduler $\sigma$ of \mdl, the probability of \mdl satisfying $\psi$ when using scheduler $\sigma$ and starting from state $s_1$ is greater or equal to the probability of \mdl satisfying $\psi$ when using scheduler $\sigma$ and starting from state $s_2$:
$$ s_1 \mosa s_2 \implies \forall \sigma \in Sch_{\mdl}^{min} \; \; Pr_{\mdl(s_1)}^{\sigma}(\psi) \geq Pr_{\mdl(s_2)}^{\sigma}(\psi) $$
\end{definition}

\begin{definition}[Assumption of Monotonic Safety for LSS]\label{def:mosLSS}
An \emph{MoS assumption for LSS} is a partial order, \mosaSMC, over the states $S$ with respect to a given safety property $\psi$ and PA \mdl. It stipulates that if two states $s_1, s_2 \in S$ are ordered, $s_1 \mosaSMC s_2$, then for every scheduler $\sigma$ of \mdl, the probability of \mdl satisfying $\psi$ when using scheduler $\sigma$ and starting from state $s_1$ is greater or equal to the probability of \mdl satisfying $\psi$ when using scheduler $\sigma$ and starting from state $s_2$:
$$ s_1 \mosaSMC s_2 \implies \forall \sigma \in Sch_{\mdl} \; \; Pr_{\mdl(s_1)}^{\sigma}(\psi) \geq Pr_{\mdl(s_2)}^{\sigma}(\psi) $$
\end{definition}

These definitions are very general. The partial orders could relate exactly two states in $S$, or they could order some subset of states in $S$, or they could even be a full ordering over states in $S$.

Such assumptions are to be made by the engineers familiar with the domain and the system at hand. By formalizing an MoS assumption, an engineer aims to integrate their high-level knowledge of the system's operation into the formal analyses of this system. Several assumptions can be made for a given system, and if they do not contradict each other they can be combined into a single overarching partial order.  

One way to find candidates for MoS assumptions is to examine the robustness, in terms of Signal Temporal Logic (STL)~\cite{stlrobustness}, of the safety property $\psi$. Intuitively, if state $s_1$ is more STL-robust with respect to $\psi$ than state $s_2$, then it may have a higher chance of remaining safe in the subsequent execution. So one can formulate an MoS partial order based on the direction of STL robustness. However, MoS assumptions can draw on substantially broader intuitions than what can be gleaned from the safety property.

For our case studies, we consider two intuitive MoS assumptions for our emergency braking system and its safety property \oursafeprop, and one intuitive MoS rule for our water tank system and its safety property \tanksafeprop. In all three rules, consider model $\mdl$ with different initial states $s_1$ and $s_2$:
\begin{itemize}
\item Distance MoS \mosdist: \emph{``Increasing the distance to the obstacle does not reduce the safety chance''}. That is, if $\mdl_1$ starts further away, i.e., $d_{s_1} \geq d_{s_2}$, then it is safer than $\mdl_2$:  $Pr_{\mdl(s_1)}(\oursafeprop) \geq Pr_{\mdl(s_2)}(\oursafeprop)$.

\item Speed MoS \mosspeed:\emph{``Decreasing the velocity does not reduce the safety chance''}. That is, if $\mdl_1$ starts slower, i.e., $v_{s_1} \leq v_{s_2}$, then it is safer than $\mdl_2$:  $Pr_{\mdl(s_1)}(\oursafeprop) \geq Pr_{\mdl(s_2)}(\oursafeprop)$.
\item Water level MoS \moswl: \emph{``Moving the water level in any tank towards its middle does not reduce safety chance''}.
This partial order has two branches, for each half of the tank. First, if $\wl_{s_2} \geq \wl_{s_1} \geq \frac{\tanksize}{2}$, then $P_{\mdl(s_1)}(\tanksafeprop) \geq P_{\mdl(s_2)}(\tanksafeprop)$. Second, if $\wl_{s_2} \leq \wl_{s_1}  \leq \frac{\tanksize}{2}$, then $P_{\mdl(s_1)}(\tanksafeprop) \geq P_{\mdl(s_2)}(\tanksafeprop)$.
\end{itemize}
Note that assumptions \mosdist and \moswl order states in the higher-robustness direction of \oursafeprop and \tanksafeprop respectively. But the second AEBS assumption, \mosspeed, is an example of a more general rule not tied to \oursafeprop. It was discovered by observing that higher speeds lead to greater distance losses, and if greater distances are to be considered safer, so should be the lower speeds. Before describing how to exploit MoS assumptions to reduce non-determinism in PAs in \Cref{sec:approach}, we present some simple counter examples showing that \mosdist, \mosspeed, and \moswl do not always holds.

\subsection{Violations of Monotonic Safety in the Motivating Systems}\label{sec:cexs}

For our systems, the introduced MoS assumptions rules appear natural and intuitive: to avoid a collision, being further away and driving slower seems safer; to avoid overflowing and underflowing, it appears safer to keep the tanks half-full. As we will experimentally show in \Cref{sec:evaluation}, these intuitions are mostly correct. However, our assumptions may not hold in some pairs of states, which we now demonstrate on simplified models. The proofs of these counterexamples are provided in \Cref{app:cexsAEBS} and \Cref{app:cexsWT}.

When discussing MoS scenarios in concrete systems, we consider \emph{state shifts} encoded as changes to states with a state vector $\Delta$. AEBS has two-dimensional states $(d, v)$ and, hence, two types of shifts corresponding to \mosdist and \mosspeed: $(\Delta_d > 0,0)$ and $(0, \Delta_v < 0)$. Each water tank is unidimensional, so we consider scalar shifts $\Delta_\wl$, such that $\Delta_\wl>0$ iff  $\wl \leq \frac{\tanksize}{2}$; otherwise, $\Delta_\wl<0$ . 

\begin{definition}\label{def:dist-indep}
An AEBS system has \emph{distance-independent perception} if the detection chance is constant: $\forall d, \detprob(d) = q$. Otherwise, it is \emph{distance-dependent}. 
\end{definition}

\begin{definition}\label{def:onebp}
An AEBS system  has \emph{one BP} if the controller bounds in \Cref{sec:system} are infinite: $C_1 = C_2 = +\infty$. Otherwise, it has \emph{multiple BPs}.  
\end{definition}

With two counterexamples for AEBS, we demonstrate that multiple BPs or distance-dependent perception alone are sufficient to violate assumptions \mosdist and \mosspeed. In both AEBS counterexamples, we set the time step $\timestep = 1s$ and the filter window size $N_{F}=1$.

The first counterexample shows that starting closer to the obstacle may be safer due to higher detection chances at closer distances. 

\begin{counterexample}\label{th:ce1}
Assumption \mosdist does not hold in the AEBS with one BP and distance-dependent perception for shift $\Delta = (1 m, 0)$ when $B = 10 m/s^2, v_0 = 11 m/s, d_0 = 13m,$ and $\detprob(d) = 1 - \lceil d\rceil/20$.
\end{counterexample}

The second counterexample shows that going at a faster speed may be beneficial because it leads to closer distances, which result in stronger braking.

\begin{counterexample}\label{th:ce2}
Assumption \mosspeed does not hold in the AEBS  with multiple BPs and distance-independent perception for shift $\Delta = (0, -1m/s)$ when $B(d) = [ 10 m/s^2\text{ if } d\leq 11m; 3 m/s^2\text{ otherwise}]$, $v_0 = 9 m/s$, $d_0 = 20m,$ and $\detprob = 0.5$.
\end{counterexample}

Intuitively, the above counterexamples arise when, due to distance or speed shifting, the model crosses a boundary into a different braking mode or a detection bin with a different probability.

\medskip

In the water tank system, the safest state is almost never $\frac{\tanksize}{2}$ because of the asymmetry in \inflow vs. \outflow, controller choices, perception errors. Therefore, it is straightforward to upset the safety balance by picking disproportionate values for inflow and outflow.

\begin{definition}\label{def:randper}
A water tank system  has \emph{random perception} if \errwin = 0 and $\forall i, \wlp_i \in \{ \emptytank, \fulltank \}$.
\end{definition}

\begin{counterexample}\label{th:ce-tank}
Assumption \moswl does not hold for a water tank system with $\tankcount = 1, \tanksize = 100, \outflow=3, \inflow=40, \timebound = 4, Pr(\wlp = 0) = 0.4,$ and $Pr(\wlp = \tanksize) = 0.6, \wl=10,$ and $\Delta_\wl = 30$. 
\end{counterexample}

\section{MoS-based Trimming}
\label{sec:approach}

This section describes how we use the MoS assumptions from \Cref{sec:assn} to remove non-determinism from PAs to speed up PMC and LSS. The two trimming procedures are slightly different, due to the differences in how PMC and LSS handle non-determinism. 
Finally, recall that our MoS trimmmings are only applied to the non-probabilistic transitions of the \mba models.




\subsection{MoS Abstraction for Model Checking}
\label{sec:MoSForMC}

We describe how to remove non-determinism from a single state as follows:

\begin{definition}[Single State PMC Transition Trimming]\label{def:MoSMCSingleStateTrimming}
    Let PA $\mdl = \{ S,\bar{s},\alpha,\delta,L \}$ with MoS assumption \mosa be given. Consider $s\in S$ and actions $\alpha_1,\alpha_2 \in \alpha$. If there exist states $s_1$ and $s_2$ such that  $(s,\alpha_1,\eta_{s_1}),(s,\alpha_2,\eta_{s_2})\in \delta$ and $s_1 \mosa s_2$, then we return PA $\mdl'$, which is identical to \mdl except that it does not have $(s,\alpha_1,\eta_{s_1})$ in its transition relation.
\end{definition}

\begin{remark}
    Another way to think about this trimming is that we remove every scheduler that picks $\alpha_1$ in state $s$ from $Sch_{\mdl'}$.
\end{remark}


We now prove that under the MoS assumption from \Cref{def:mosPMC} each application of the trimming procedure in \Cref{def:MoSMCSingleStateTrimming} is conservative. First, we introduce the following lemma which states that the probability of $\mdl$ satisfying safety property $\psi$ can be decomposed by reasoning about the paths of $\mdl$ which do and do not pass through some state $s$. The proof can be found in \Cref{lem:traceDecompProof}.


\begin{lemma}[Satisfaction Probability by Trace Decomposition]\label{lemma:traceDecomp}
Let $\mdl$ be a PA. Assume that there are no infinite paths through \mdl. Let $\psi$ be an LTL safety property. Let $\sigma$ be a scheduler of $\mdl$ and let $s$ be a state of \mdl. Then 
\begin{align*}
    Pr_{\mdl}^{\sigma}(\psi) &=
    Pr_{\mdl}^{\sigma}(\psi \& \neg \diamond s) + Pr_{\mdl}^{\sigma}(\psi \& \diamond s) \\
    &= Pr_{\mdl}^{\sigma}(\psi \& \neg \diamond s) + Pr_{\mdl}^{\sigma}(\diamond s) Pr_{\mdl(s)}^{\sigma}(\psi)
\end{align*}
\end{lemma}

We now state the single state trimming conservatism result and its proof.

\begin{theorem}[PMC MoS and PMC Transition Trimming Implies PMC Conservatism]\label{thm:MCCons}
Let $\mdl$ $= \{ S,\bar{s},\alpha,\delta,L \}$ be a PA with MoS assumption \mosa. Assume that there are no infinite paths through \mdl. Let $\psi$ be an LTL safety property. If we trim $\mdl$ according to the procedure detailed in \Cref{def:MoSMCSingleStateTrimming} for state $s$ to arrive at model $\mdl'$, 
then $Pr_{\mdl}^{min}(\psi) = Pr_{\mdl'}^{min}(\psi)$.
\end{theorem}
\begin{proof}
Assume that state $s$ has non-deterministic transitions to state $s_1$ via action $\alpha_1$ and state $s_2$ via action $\alpha_2$:  $$(s,\alpha_1,\eta_{s_1}),(s,\alpha_2,\eta_{s_2}) \in \delta$$ In addition, assume that $s_1 \mosa s_2$. Now consider PA $M'$ that is formed by removing $(s,\alpha_1,\eta_{s_1})$ from $\delta$. 
To prove that the trimming is conservative, we just need to show that at least one min scheduler of \mdl for $\psi$ is contained in the set of schedulers for $\mdl'$.

Let $\sigma$ be a min scheduler of \mdl for $\psi$. So $Pr_{\mdl}^{\sigma}(\psi) = Pr_{\mdl}^{min}(\psi)$. We now show that either $\sigma$ or some other min scheduler is contained in the set of schedulers for $\mdl'$ by contradiction. Assume that $\sigma \notin Sch_{\mdl'}$. By the strong MoS assumption, $\sigma$ chooses action $\alpha_1$ in state $s$. Now consider scheduler $\sigma'$ which picks action $\alpha_2$ in state $s$ and is otherwise identical to $\sigma$. By \Cref{lemma:traceDecomp}, we have that:
\begin{align*}
    Pr_{\mdl}^{\sigma}(\psi) &= Pr_{\mdl}^{\sigma}(\psi \& \neg \diamond s) + Pr_{\mdl}^{\sigma}(\diamond s) Pr_{\mdl(s)}^{\sigma}(\psi) \\
    Pr_{\mdl}^{\sigma'}(\psi) &= Pr_{\mdl}^{\sigma'}(\psi \& \neg \diamond s) + Pr_{\mdl}^{\sigma'}(\diamond s) Pr_{\mdl(s)}^{\sigma'}(\psi)
\end{align*}

Note that $\sigma$, $\sigma'$ differ only in state $s$ and \mdl has no self loops since it only has finite paths. So:
$$Pr_{\mdl}^{\sigma}(\psi \& \neg \diamond s) = Pr_{\mdl}^{\sigma'}(\psi \& \neg \diamond s) \; \text{and} \; Pr_{\mdl}^{\sigma}(\diamond s) = Pr_{\mdl}^{\sigma'}(\diamond s) $$

For the final two terms, we unroll one step of each of the models from state $s$ under their respective schedulers. Scheduler $\sigma$ takes $\mdl$ to state $s_1$ and $\sigma'$ takes $\mdl$ to state $s_2$. Since $\sigma$, $\sigma'$ differ only in state $s$ and \mdl cannot loop back to $s$ from $s_2$, it follows that $Pr_{\mdl(s_2)}^{\sigma'}(\psi)=Pr_{\mdl(s_2)}^{\sigma}(\psi)$.

Since $s_1 \mosa s_2$ we have that:
\begin{align*}
    Pr_{\mdl(s_1)}^{\sigma}(\psi) \geq Pr_{\mdl(s_2)}^{\sigma}(\psi) = Pr_{\mdl(s_2)}^{\sigma'}(\psi)
\end{align*}

It follows that $Pr_{\mdl}^{\sigma}(\psi) \geq Pr_{\mdl}^{\sigma'}(\psi)$. If the inequality is strict then we have a contradiction, as $\sigma$ is not a min scheduler of \mdl and if the terms are equal then both $\sigma$ and $\sigma'$ are both min schedulers and the MoS trimming preserves at least one.

\end{proof}

The model-wide PMC MoS trimming procedure applies the trimming procedure in \Cref{def:MoSMCSingleStateTrimming} to all states of the model.

\begin{corollary}\label{cor:PMCConsModelWide}
    The model-wide PMC MoS trimming procedure is conservative under \Cref{def:mosPMC}. 
\end{corollary}

\subsection{MoS Abstraction for Statistical Model Checking}
\label{sec:MoSForSMC}

We describe how to remove non-determinism from a single state for LSS as follows:

\begin{definition}[Single State LSS Transition Trimming]\label{def:MoSSMCSingleStateTrimming}
    Let $\mdl = \{ S,\bar{s},\alpha,\delta,L \}$ be a PA with MoS assumption \mosaSMC. Consider $s\in S$. Let $\alpha_1,\hdots,\alpha_d$ be the actions enabled in $s$. So    $$(s,\alpha_1,\eta_{s_1}),\hdots,(s,\alpha_d,\eta_{s_d}) \in \delta$$ for some $s_1,\hdots,s_d \in S$. If 
    $$s_1 \mosaSMC s_d,\hdots s_{d-1} \mosaSMC s_d$$
    then we return PA $\mdl'$ which is identical to \mdl except that it does not have $(s,\alpha_1,\eta_{s_2}),\hdots,(s,\alpha_{d-1},\eta_{s_{n-1}})$ in its transition relation.
\end{definition}

\begin{remark}
    Another way to think about this trimming is that we remove every scheduler that picks $\alpha_1,\hdots,\alpha_{d-1}$ in state $s$ from $Sch_{\mdl'}$.
\end{remark}


We now prove that under the MoS assumption from \Cref{def:mosLSS} each application of the trimming procedure in \Cref{def:MoSSMCSingleStateTrimming} is conservative. Since the LSS outputs are random variables, when we say that the results are conservative we mean that the LSS results on the untrimmed model have first order stochastic dominance (FSD) over those of the trimmed model. This is a formal way of defining that we expect \mdl to return larger safety chances than $\mdl'$.

\begin{definition}[First-order Stochastic Dominance (FSD)]\label{def:fsd}
We say that random variable $A$ has FSD over random variable $B$ if $F_{A}(x) \leq F_{B}(x) \; \forall x$, where $F_A(x)$ and $F_B(x)$ are the cumulative density functions (CDFs) of $A$ and $B$ at value $x$. We denote this as $A \geq_{FSD} B$.
\end{definition}

\begin{remark}
    One way to think about this definition is that $A \geq_{FSD} B$ if the probability mass of $A$ lies to the right of that of $B$.
\end{remark}

We now state the single state conservatism result and give its proof.

\begin{theorem}[LSS MoS and LSS MoS Trimming Implies LSS Conservatism]\label{thm:SMCCons}

Let \mdl be a PA with MoS assumption \mosaSMC. Assume that there are no infinite paths through \mdl. Let $\psi$ be an LTL safety property. If we trim $\mdl$ according to the procedure detailed in \Cref{def:MoSSMCSingleStateTrimming} for state $s$ to arrive at model $\mdl'$, 
then $Pr_{\mdl}^{LSS,n}(\psi) \geq_{FSD} Pr_{\mdl'}^{LSS,n}(\psi)$.

\end{theorem}

\begin{proof}
Let $s$ be a state in $\mdl$. Assume that state $s$ has non-deterministic transitions to states $s_1,s_2,\hdots,s_d$  $$(s,\alpha_1,\eta_{s_1}),\hdots,(s,\alpha_d,\eta_{s_d})\in \delta$$ and that:
$$s_1 \mosa s_d, \hdots, s_{d-1} \mosa s_d$$
Now consider PA $\mdl'$ formed by removing $(s,\alpha_1,\eta_{s_1}),\hdots,(s,\alpha_d,\eta_{s_{d-1}})$ from $\delta$.



Consider scheduler $\sigma \in Sch_{\mdl}$. Now we generate its corresponding scheduler in $\mdl'$, called $\sigma'$, by requiring $\sigma'$ to take action $\alpha_d$ to state $s_d$ whenever it is in state $s$. Effectively, we are saying that of the $d$ choices that $\sigma$ could make in state $s$, $\sigma'$ takes $\alpha_d$. So each $\sigma' \in Sch_{\mdl'}$ has $d$ corresponding schedulers in $Sch_{\mdl}$. It follows that $|Sch_{\mdl}| = d*|Sch_{\mdl'}|$. 

Now we show that $Pr_{\mdl}^{\sigma}(\psi) \geq Pr_{\mdl'}^{\sigma'}(\psi)$. By \Cref{lemma:traceDecomp}, we have:
\begin{align*}
    Pr_{\mdl}^{\sigma}(\psi) &= Pr_{\mdl}^{\sigma}(\psi \& \neg \diamond s) + Pr_{\mdl}^{\sigma}(\diamond s) Pr_{\mdl(s)}^{\sigma}(\psi) \\
    Pr_{\mdl'}^{\sigma'}(\psi) &= Pr_{\mdl'}^{\sigma'}(\psi \& \neg \diamond s) + Pr_{\mdl'}^{\sigma'}(\diamond s) Pr_{\mdl'(s)}^{\sigma'}(\psi)
\end{align*}

Note that $\sigma$, $\sigma'$ differ only in state $s$ and \mdl has no self loops since it only has finite paths. So:
$$Pr_{\mdl}^{\sigma}(\psi \& \neg \diamond s) = Pr_{\mdl'}^{\sigma'}(\psi \& \neg \diamond s) \; \text{and} \; Pr_{\mdl}^{\sigma}(\diamond s) = Pr_{\mdl'}^{\sigma'}(\diamond s) $$

For the final two terms, we unroll one step of the models from state $s$. This one step takes model $\mdl$ to one of the following states: $\{ s_1,s_2,\hdots,s_d\}$ (determined by scheduler $\sigma$) and model $\mdl'$ to state $s_d$ (since $\mdl'$ doesn't have non-deterministic transitions to $s_1,\hdots,s_{d-1}$). By \mosaSMC, $Pr_{\mdl(s_d)}^{\sigma}(\psi) \leq Pr_{\mdl(s_i)}^{\sigma}(\psi), \; \; i=1,\hdots,d-1$. Finally, since $\mdl$ and $\mdl'$ and $\sigma$ and $\sigma'$ differ only in state $s$ and \mdl cannot loop back to $s$ from $s_2,\hdots,s_{d-1}$, it follows that $Pr_{\mdl(s_d)}^{\sigma}(\psi) = Pr_{\mdl'(s_d)}^{\sigma'}(\psi)$. So $Pr_{\mdl}^{\sigma}(\psi) \geq Pr_{\mdl'}^{\sigma'}(\psi)$.

Now draw some $\sigma_1 \sim Unif(Sch_{\mdl})$. Let $\sigma_1'$ be the corresponding scheduler in $sch_{\mdl'}$. It follows that $\sigma_1' \sim Unif(sch_{\mdl'})$:
\begin{align*}
    P(\sigma_1' = \sigma') = P(\text{$\sigma_1$ corresponds to $\sigma'$})
    = \frac{1}{|Sch_{\mdl}|} d
    = \frac{1}{|Sch_{\mdl'}|}
\end{align*}

Now we apply LSS simultaneously to \mdl and $\mdl'$. 

For \mdl, we pick $n$ schedulers uniformly $\sigma_1,\hdots,\sigma_n \sim Unif(Sch_{\mdl})$, compute each scheduler's safety probability $p_i := Pr_{\mdl}^{\sigma_i}(\psi)$ and return the smallest safety probability: $min_{i=1,\hdots,n}(p_i)$. 


For $\mdl'$, we convert each $\sigma_i$ into its corresponding $\sigma_{i}' \in Sch_{\mdl'}$, compute its safety probability $p'_i := Pr_{\mdl'}^{\sigma'_i}(\psi)$ and return the minimum one: $min_{i=1,\hdots,n}(p'_i)$.



From earlier, we know that $p_i \geq p'_i \; \forall i=1,\hdots,n$. It follows that $min_{i=1,\hdots,n}(p_i) \geq min_{i=1,\hdots,n}(p'_i)$. So with this shared sample space of schedulers we always get a larger value from running LSS on $\mdl$ than $\mdl'$. Thus, the probability mass of $Pr_{\mdl}^{LSS}(\psi)$ lies to the right of that of $Pr_{\mdl'}^{LSS}(\psi)$. Thus, $Pr_{\mdl}^{LSS}(\psi) \geq_{FSD} Pr_{\mdl'}^{LSS}(\psi)$.


\end{proof}
The model-wide LSS MoS trimming procedure works by applying the trimming procedure in \Cref{def:MoSSMCSingleStateTrimming} to all states of the model.

\begin{corollary}\label{cor:LSSConsModelWide}
    The model-wide LSS MoS trimming procedure is conservative under \Cref{def:mosLSS}. 
\end{corollary}

\subsubsection{LSS Data Efficiency}
\label{sec:MoSLSSDataEff}

As LSS is a sampling based method, it has a trade-off between sample count and accuracy. Sampling more schedulers produces more accurate LSS results. Intuitively, the larger the set of schedulers for a PA \mdl, the more schedulers LSS should sample to get a reasonable probability estimate of $Pr_{\mdl}^{min}(\psi)$. Thus, reducing the number of schedulers when constructing $\mdl'$ allows LSS to sample fewer schedulers. In our evaluation, we show that our trimmings allow for better LSS results while using 10x fewer schedulers than without trimming.

\section{Evaluation and Results} 
\label{sec:evaluation}

Our experimental evaluation had two goals: compare the scalability/data efficiency of our abstractions and explore the extent to which \Cref{def:mosLSS} holds on small models\footnote{Large models have too many schedulers to enumerate them.}. We perform evaluation on simulations of the two case studies described in \Cref{sec:system}. For PMC we use the PRISM model checker \cite{Kwiatkowska2011PRISM} and for LSS we use the MODEST model checker \cite{Budde2018RareLSS}. The code and models used can be found at: \texttt{https://github.com/earnedkibbles58/ICCPS\_2022\_MoS}.

\subsection{Setup: Data Collection and Modeling}

\subsubsection{AEBS}
First, we setup the AEBS controller in the self-driving car simulator CARLA~\cite{Dosovitskiy17}, using a red Toyota Prius as the obstacle. We performed two data collections: (i) constant-$v$ runs to gather image-distance data to construct PA \mper, (ii) AEBS-controlled runs to approximate the true value of \oursafeprob.

\looseness=-1
For the former collection, we ran 500 simulations starting from 200m and obtained 180500 pairs ($d$, $\outlabs$). Model \mper contains the probabilities of detection conditioned on the 10m-wide distance bin and the past 3 low-level outcomes. For the latter collection, we ran 300 simulations of emergency braking starting at 160m and 20m/s, with YoloNet perception and a 3-wide majority filter on top. 20 runs out of 300 resulted in a crash, yielding a $[0.04, 0.1]$ 95\% confidence interval (CI) for the true collision chance (i.e., $1 - \oursafeprob$).

\subsubsection{Water Tank}
To derive a probabilistic model of the simulated perception error with $\errwin=6$, at each water level (from $1$ to \tanksize) we ran 100 trials of the perception and recorded the counts of perception errors into $15$ bins, $13$ of which were of width $1$, representing errors from $-6$ to $6$. Errors above $6$ or below $-6$ got their own bins as well, which represented the perception giving an output of $0$ or \tanksize, respectively. To summarize, the model of \mper was a categorical distribution over the perception error with the following bins:
\begin{align*}
    \{0,wl-6,wl-5,\hdots,wl+5,wl+6,\tanksize\}
\end{align*}

\subsection{PMC Scalability and Accuracy}\label{sec:results}

For each case study, we consider two different models: (i) \mba from \Cref{sec:modelBackground}, (ii) $\mbatt^{PMC}$, which we form by applying the MoS-driven trimming procedures from \Cref{def:MoSMCSingleStateTrimming}.


\subsubsection{AEBS}
\looseness=-1
One experiment evaluates the accuracy-scalability tradeoffs for a fixed initial state (160m, 20m/s) and differing distance interval sizes of the models. \Cref{fig:AEBSDiffDistDiscs} shows that $\mbatt^{PMC}$ outperforms \mba in terms of run time while maintaining empirical conservatism. This shows that our MoS-enabled trimming shows significant speed ups while maintaining empirical conservatism, as it returns a lower safety chance than the simulations, for a range of model hyperparameters.


\begin{figure}
    \centering
    \includegraphics[width=1\linewidth]{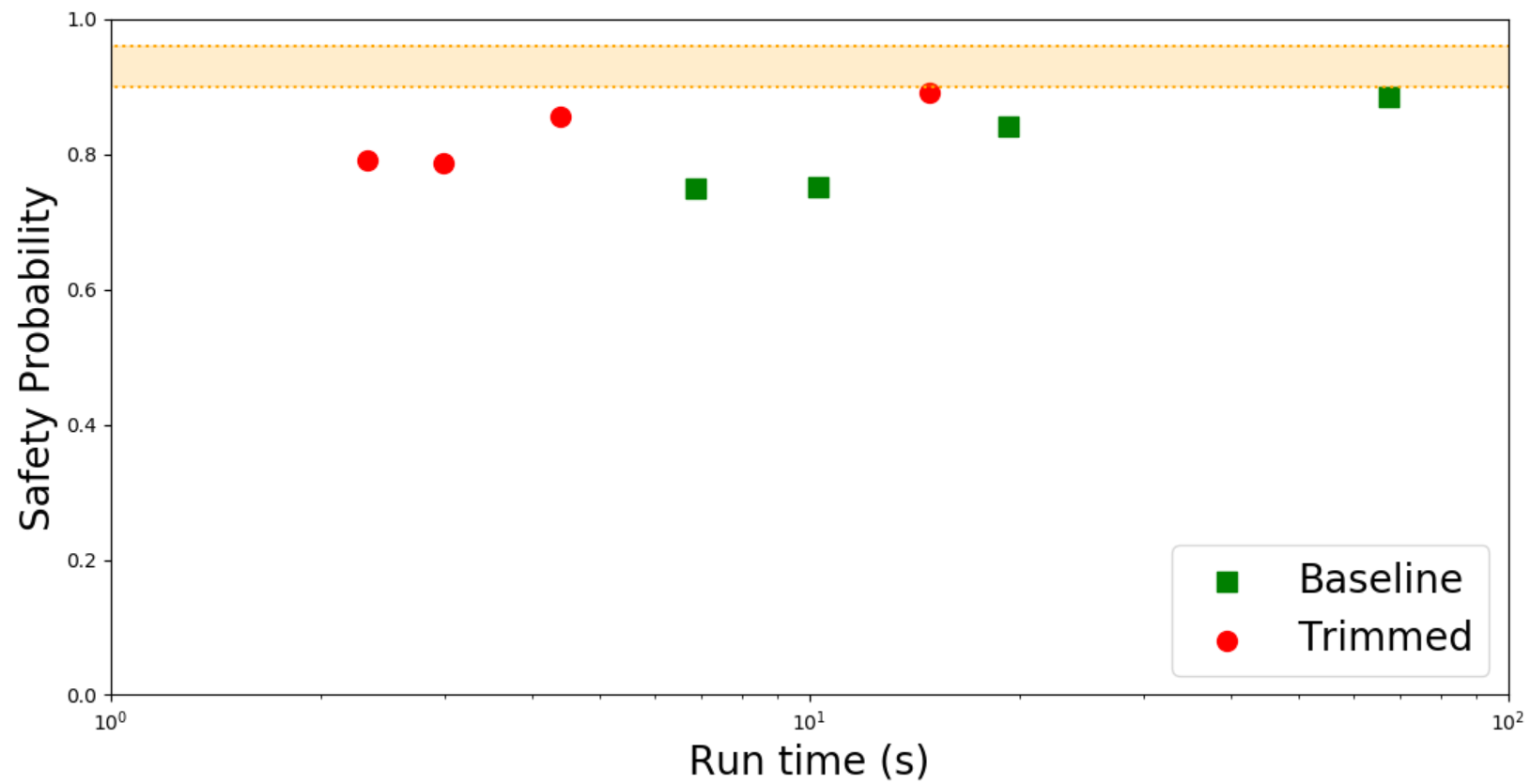}
    \caption{\small AEBS run times and safety probabilities for \mba (green squares) and $\mbatt^{PMC}$ (red dots),  and 95\% CI for the simulated collision chance (orange band). Each dot represents a different interval size for \mba.}
    \label{fig:AEBSDiffDistDiscs}
    \vspace*{-3mm}
\end{figure}


\begin{figure}
\centering
\begin{subfigure}{0.25\textwidth}
  \centering
  \includegraphics[width=1\linewidth]{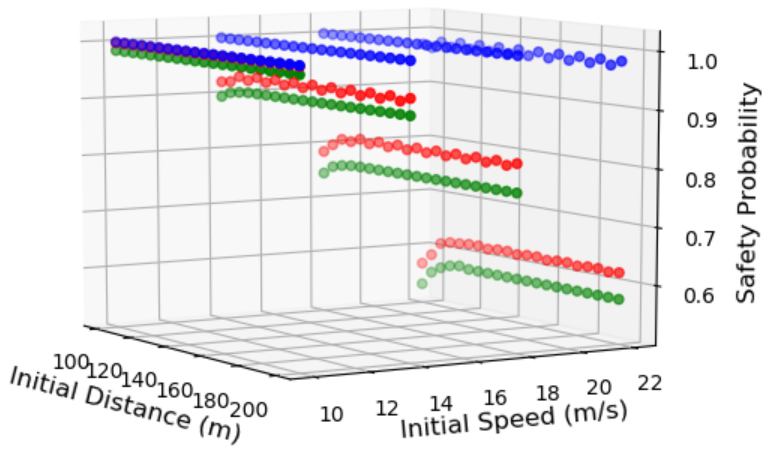}
  \label{fig:scalabilityCrashProbsAEBS}
\end{subfigure}%
\begin{subfigure}{0.25\textwidth}
  \centering
  \includegraphics[width=1\linewidth]{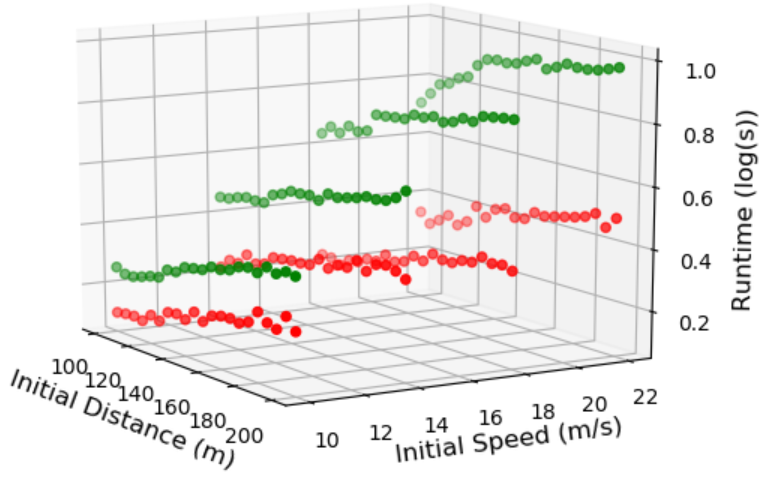}
  \label{fig:scalabilityRuntimesAEBS}
\end{subfigure}
\caption{\small Comparison of PMC safety probabilities (left) and run times (right) for models across different initial distances and velocities of \mba (green), $\mbatt^{PMC}$ (red) and $\mbattneg^{PMC}$ (blue) models of the AEBS example.}
\label{fig:scalabilityAEBS}
\vspace*{-3mm}
\end{figure}

In the other experiment we fixed the distance interval size and varied the initial conditions, shown in \Cref{fig:scalabilityAEBS}. First, \mba scales very poorly. It takes an order of magnitude longer to run than $\mbatt^{PMC}$. Second, $\mbatt^{PMC}$ returns similar safety probabilities to \mba. Third, we computed a model using the negations of \mosdist and \mosspeed, which we denote as $\mbattneg^{PMC}$. This model returns very large safety probabilities. This is further evidence that \mosdist and \mosspeed are fairly accurate, as taking their negations produces overly safe results. 


\subsubsection{Water Tank}
We first fixed the initial water levels to be $50$ (exactly half full) and varied the water level interval size of the models. \Cref{fig:waterTankDiffWLDiscs} shows that the $\mbatt^{PMC}$ outperforms \mba in terms of run time while returning similar safety probabilities over a range of hyperparameters. 

We then fixed the water level interval size and varied the initial conditions, shown in \Cref{fig:scalabilityWaterTank}. First, \mba scales very poorly, with several models taking upwards of an hour to verify. $\mbatt^{PMC}$ speeds this by up to 2 orders of magnitude. Second, $\mbatt^{PMC}$ returns very similar safety probabilities to \mba. Third, \mbattneg returns really high safety probabilities, which further indicates that \moswl is an accurate assumption.


\begin{figure}
    \centering
    \includegraphics[width=0.8\linewidth]{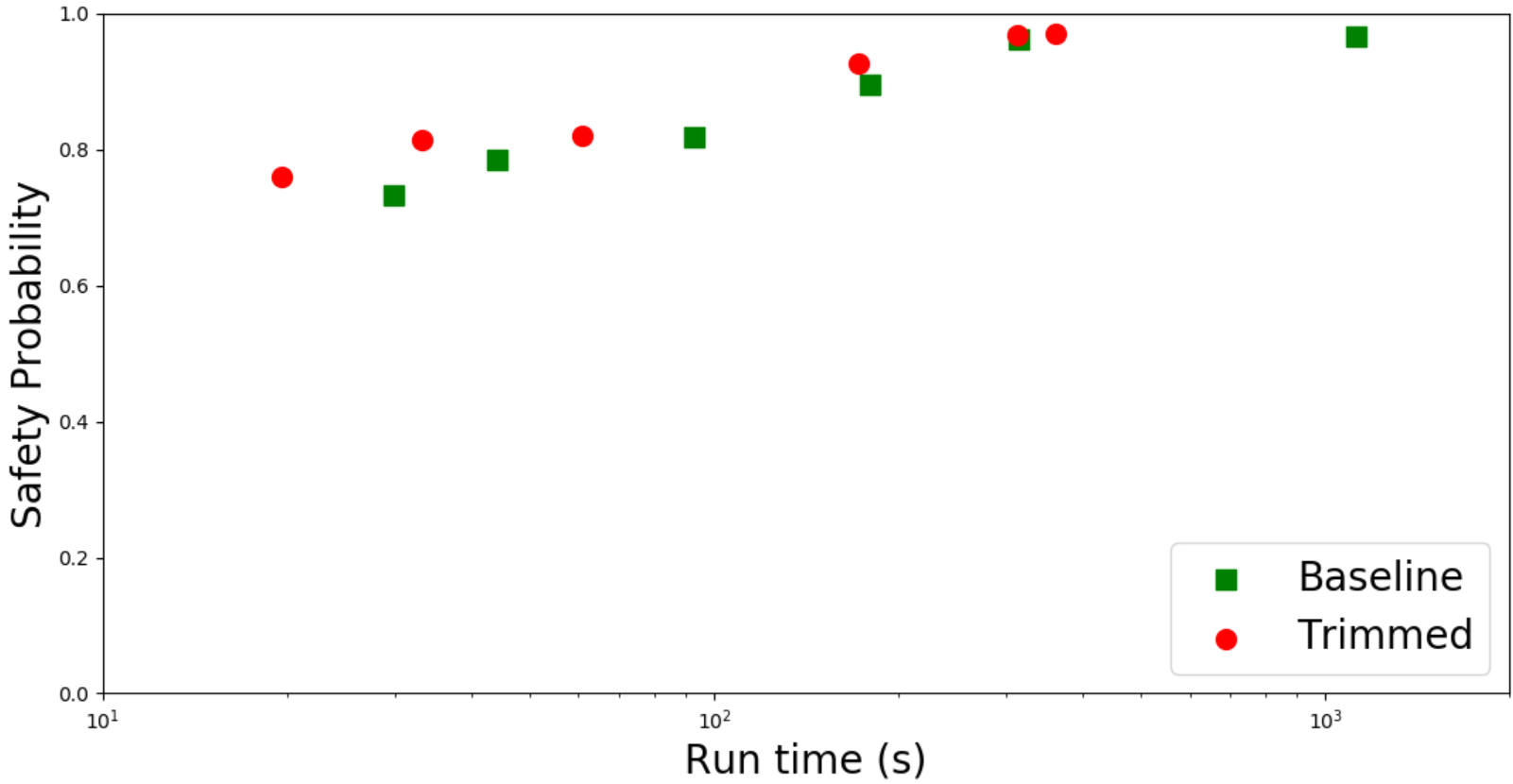}
    \caption{\small Water tank PMC run times and overflow/underflow probabilities for \mba (green square) and $\mbatt^{PMC}$ (red dot). Each dot is for a set of values of a model's hyperparameters.}
    \label{fig:waterTankDiffWLDiscs}
\end{figure}




\begin{figure}
\centering
\begin{subfigure}{0.25\textwidth}
  \centering
  \includegraphics[width=1\linewidth]{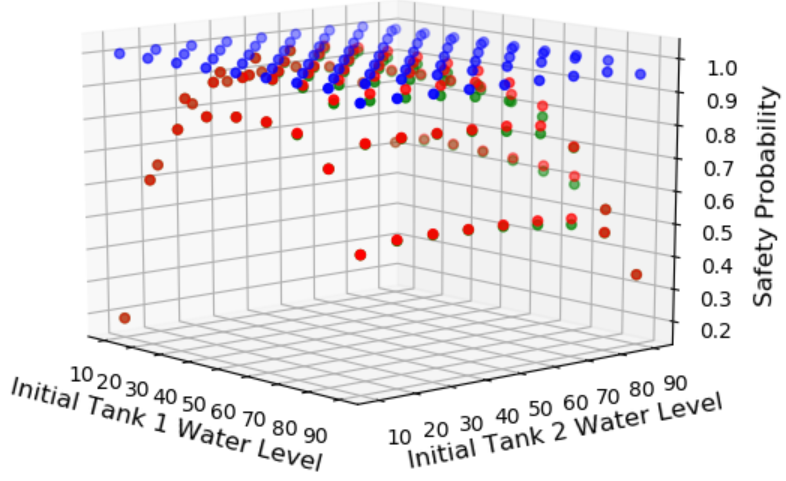}
  \label{fig:scalabilityCrashProbsWaterTank}
\end{subfigure}%
\begin{subfigure}{0.25\textwidth}
  \centering
  \includegraphics[width=1\linewidth]{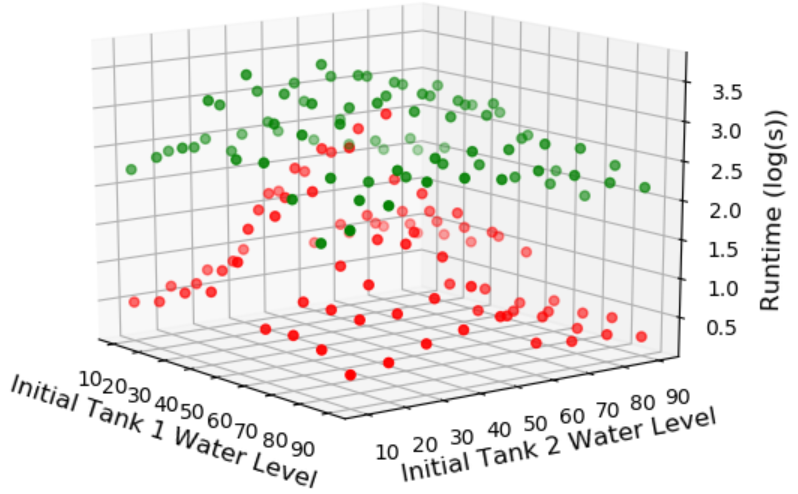}
  \label{fig:scalabilityRuntimesWaterTank}
\end{subfigure}
\caption{\small Comparison of PMC safety probabilities (left) and run times (right) of \mba (green), $\mbatt^{PMC}$ (red) and $\mbattneg^{PMC}$ (blue) models for the water tank example.}
\label{fig:scalabilityWaterTank}
\vspace*{-3mm}
\end{figure}

\subsection{LSS Data Efficiency}

For each case study, we consider two different models: (i) \mba from \Cref{sec:modelBackground}, (ii) $\mbatt^{LSS}$, which we form by applying the MoS-driven trimming procedure from \Cref{def:MoSSMCSingleStateTrimming}.


We kept the same distance and water level interval sizes as before and ran LSS on a smaller range of initial conditions for both abstractions. The safety probability of each scheduler was computed to an absolute error of $0.05$ with $0.8$ confidence. We ran LSS using $10$ scheduler samples on both \mba and $\mbatt^{LSS}$. In addition, we ran LSS using $1$ scheduler sample on $\mbatt^{LSS}$. For all cases, we ran $10$ different trials of LSS and averaged the results. The AEBS results are shown in \Cref{tab:AEBSLSSProbs} and \Cref{tab:AEBSLSSRunTime} 
and the water tank results are shown in \Cref{fig:scalabilityWaterTankLSS}. The $\mbatt^{LSS}$ models give more conservative safety probabilities with 10x fewer schedulers and, hence, 10x smaller run times. This is because all the extra non-determinism in the \mba models dilutes the scheduler sampling. Note that $\mbatt^{LSS}$ took around the same time as \mba when using an equal number of samples.

\begin{table}
\centering
\begin{tabular}{|c|c|c|c|}
    \hline \specialcell{Initial \\ Condition (d,v)}   & \specialcell{\mba and \\ 10 Schedulers} & \specialcell{$\mbatt^{LSS}$ and \\ 10 Schedulers} & \specialcell{$\mbatt^{LSS}$ and \\ 1 Scheduler}  \\
    \hline (130,14) & 0.991 & 0.919 & 0.962 \\
    \hline (130,18) & 0.962 & 0.802 & 0.831 \\
    \hline (160,14) & 0.993 & 0.919 & 0.943 \\
    \hline (160,18) & 0.966 & 0.802 & 0.846 \\
    \hline
\end{tabular}
\caption{LSS Safety Probabilities for the AEBS system.}
\label{tab:AEBSLSSProbs}
\vspace*{-3mm}
\end{table}

\begin{table}
\centering
\begin{tabular}{|c|c|c|c|}
    \hline \specialcell{Initial \\ Condition (d,v)}   & \specialcell{\mba and \\ 10 Schedulers } & \specialcell{$\mbatt^{LSS}$ and \\ 10 Schedulers} & \specialcell{$\mbatt^{LSS}$ and \\ 1 Scheduler}  \\
    \hline (130,14) & 171 & 201 & 15 \\
    \hline (130,18) & 603 & 179 & 11 \\
    \hline (160,14) & 203 & 250 & 17 \\
    \hline (160,18) & 608 & 228 & 12 \\
    \hline
\end{tabular}
\caption{LSS Run Times (sec) for the AEBS system.}
\label{tab:AEBSLSSRunTime}
\vspace*{-3mm}
\end{table}

\begin{figure}
\centering
\begin{subfigure}{0.25\textwidth}
  \centering
  \includegraphics[width=1\linewidth]{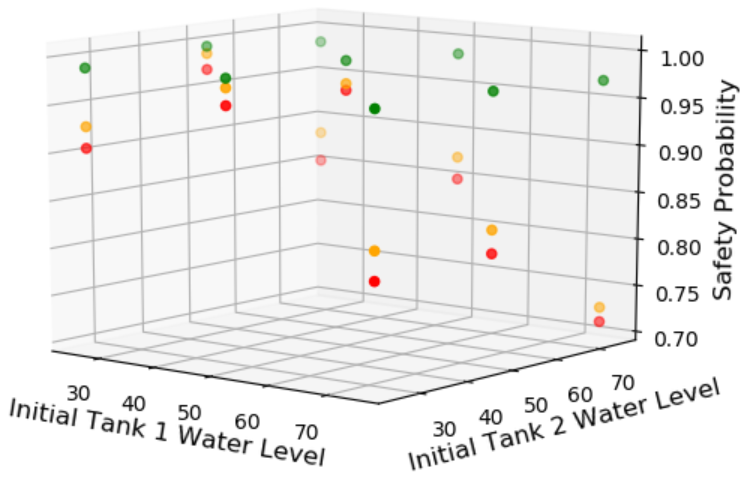}
  \label{fig:scalabilityCrashProbsWaterTankLSS}
\end{subfigure}%
\begin{subfigure}{0.25\textwidth}
  \centering
  \includegraphics[width=1\linewidth]{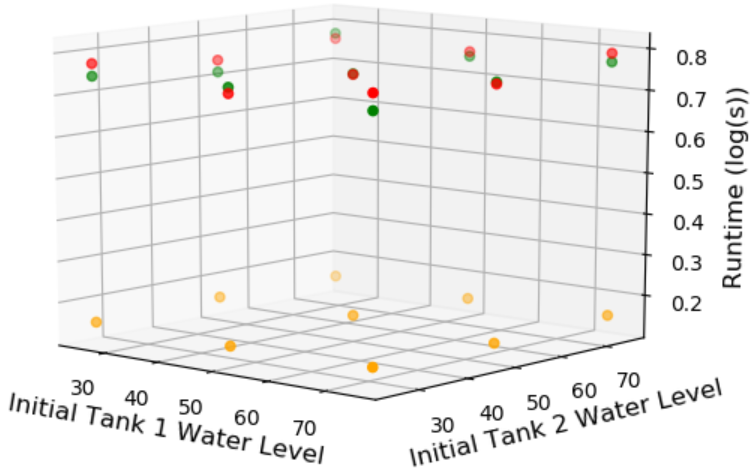}
  \label{fig:scalabilityRuntimesWaterTankLSS}
\end{subfigure}
\caption{\small LSS safety probabilities (left) and run times (right) of water tank models across different initial water levels for \mba and 10 schedulers (green), $\mbatt^{LSS}$ and 10 schedulers (red) and $\mbatt^{LSS}$ and 1 scheduler (orange).}
\label{fig:scalabilityWaterTankLSS}
\vspace*{-3mm}
\end{figure}

\subsection{Validating MoS Assumptions}
We generated small \mba models for the AEBS and water tanks model and enumerated every scheduler of both. We computed the proportion of schedulers for which \Cref{def:mosLSS} held for each pair of states involved in the MoS trimming, thus quantifying the extent to which our MoS assumptions held on these models. 

\subsubsection{AEBS}

We kept the same discretization parameters as before and used an initial distance of $9$m and speed of $1.2$m/s. \mba had $11664$ schedulers and a minimum safety chance of $0.986$. $\mbatt^{PMC}$ had a safety chance of $0.990$, so our MoS-trimming is not conservative and the conjunction of \mosdist and \mosspeed does not hold. However, we want to examine how close this conjunction is to holding. So we computed each state pair involved in the MoS-trimming procedure detailed in \Cref{sec:MoSForMC}, of which there were 17. Consider one such pair $(s_1,s_2)$, with $s_1 \mosa s_2$. We then computed the safety probability of the model when starting in both $s_1$ and $s_2$ for every scheduler of the model and computed the proportion of schedulers for which $s_1$ gave a higher safety chance than $s_2$. This amounts to computing $P_{\sigma \sim Sch_{\mdl}} \left( Pr_{\mdl(s_1)}^{\sigma}(\oursafeprop) \geq Pr_{\mdl(s_2)}^{\sigma}(\oursafeprop) \right)$.\footnote{If \moswl held, then this proportion would be $1$.} We denote this probability as $p_{s_1,s_2}$.

We applied this procedure to each of the $17$ trimmed state pairs. For $15$ of them, $p_{s_1,s_2}=1$. The other two had values of $0.5$ and $0.656$, respectively. So \mosdist and \mosspeed hold perfectly for the majority of the state pairs. But for the pairs where it didn't hold it was not very close to holding. This is most likely due to those two state pairs being near the AEBS control boundaries and echoes the intuitions from the counterexamples, which is that \mosdist and \mosspeed do not hold near control and perception boundaries.

\subsubsection{Water Tank}

We used a model of a single water tank with a max water level of $30$ and the same $T$ and water level interval size as before. \mba had $1024$ schedulers and a minimum safety chance of $0.261$. $\mbatt^{PMC}$ had a safety chance of $0.303$. So our MoS-trimming is not conservative and \moswl does not hold. However, we want to examine how close \moswl is to holding. So we computed each state pair involved in the MoS-trimming procedure from \Cref{sec:MoSForMC} used to convert \mba into $\mbatt^{PMC}$, of which there were 5. Consider one such pair $(s_1,s_2)$, with $s_1 \mosa s_2$. We then computed the safety probability of the model when starting in both $s_1$ and $s_2$ for every scheduler of the model and calculated the proportion of schedulers for which $s_1$ gave a higher safety chance than $s_2$. This amounts to computing $P_{\sigma \sim Sch_{\mdl}} \left( Pr_{\mdl(s_1)}^{\sigma}(\tanksafeprop) \geq Pr_{\mdl(s_2)}^{\sigma}(\tanksafeprop) \right)$. 

We applied this process for each of the $5$ state pairs and the results are shown in \Cref{tab:MoSEvalWaterTank}. This shows that although our MoS-trimming is not conservative, our MoS assumption \moswl holds for most schedulers and state pairs. We see that MoS holds over more schedulers near the extremes of the water tank. This also demonstrates why the MoS-trimming is so conservative for LSS, since most of the trimmed schedulers are in fact overly optimistic.

\begin{table}[]
    \centering
    \begin{tabular}{|c|c|c|}
        \hline State $s_1$ & State $s_2$ & $p_{s_1,s_2}$  \\
        \hline $[5,10]$ &  $[0,5]$ & $1.0$\\
        \hline $[10,15]$ &  $[5,10]$ & $0.934$  \\
        \hline $[15,20]$ & $[10,15]$ & $0.813$ \\
        \hline $[15,20]$ & $[20,25]$ & $0.881$ \\
        \hline $[20,25]$ & $[25,30]$ & $1.0$ \\
        \hline
    \end{tabular}
    \caption{Probabilities of \moswl holding over trimmed state pairs of \mba.}
    \label{tab:MoSEvalWaterTank}
    \vspace*{-8mm}
\end{table}

\section{Related Work}
\label{sec:related}

Many works have applied statistical analysis methods to probabilistic systems with non-determinism. As previously mentioned, \cite{legay2014LSS} introduced LSS and improved upon it with smart sampling \cite{DArgenio2015SmartLSS}. These LSS techniques have been implemented in the MODEST model checker \cite{hartmanns2014modest}. These methods have been applied to the analysis of timed probabilistic models in \cite{DArgenio2016TimedLSS,Hartmanns2017TimedLSS} and rare event detection \cite{Budde2020RareLSS}. Reinforcement learning techniques have also been employed to find min schedulers of markov decision processes \cite{Henriques2012}.

Previous works have addressed compositional reasoning for probabilistic models, such as extending notions of simulation and bisimulation to PAs~\cite{SegalaThesis,Segala1994,Gebler_2013}.
These relations are preserved under parallel composition, but they are too fine-grained for compositional verification. Another approach uses the notion of trace distribution inclusion (distributions over traces of automata actions) for compositional reasoning \cite{Segala1995}, but this relation is not preserved under parallel composition. Other works target more restrictive modeling formalisms than PAs, such as reactive modules \cite{DiAlfaro2001} and probabilistic I/O systems \cite{Cheung2004}. Another work \cite{Feng11} automatically generates assumptions for composition. However, it forbids the composed model from having nondeterministic transitions, which does not work with the abstractions presented in this paper. 

The most conceptually similar topic to our MoS assumptions is \textit{sensitivity analysis} of dynamical systems, which aims to quantify how changes in system initial conditions affect system traces \cite{donze2007}.


\looseness=-1
Reasoning about the safety of autonomous vehicles has been studied from other perspectives as well. The Responsibility-Sensitive Safety (RSS) model~\cite{mobileye} uses physics-based safety analysis for provable safety in different driving scenarios for autonomous vehicles. RSS does not consider ML-based perception models in the closed loop of their autonomous systems. Finally, researchers have investigated the use of falsification to find perception outputs that cause unsafe system behaviors~\cite{comp-fal}.

\section{Conclusion}
\label{sec:discussion}

This paper introduced an abstraction approach for scalable probabilistic model checking and data efficient lightweight scheduler sampling of probabilistic automata. Specifically, we use intuitive assumptions of monotonic safety to remove non-determinism from PAs. When such assumptions hold, they lead to provably conservative abstractions. Even when these assumptions do not perfectly hold, MoS-based simplifications remain empirically conservative. We demonstrated these MoS-driven trimmings on case studies of a self-driving car and a water tank. Future work includes analytically bounding conservatism loss from MoS violations, extending the MoS trimming techniques to distributions of states, and exploring how our notion of MoS relate to sensitivity analysis for dynamical systems. More generally, we see the development of abstractions suitable for LSS as a promising area to explore.

\section*{Acknowledgments}

The authors thank Ramneet Kaur for helping to develop the AEBS case study.
This work was supported in part by ARO W911NF-20-1-0080, AFRL and DARPA FA8750-18-C-0090, and ONR N00014-20-1-2744. Any opinions, findings and conclusions or recommendations expressed in this material are those of the authors and do not necessarily reflect the views of the Air Force Research Laboratory (AFRL), the Army Research Office (ARO), the Defense Advanced Research Projects Agency (DARPA), the Office of Naval Research (ONR) or the Department of Defense, or the United States Government.


\bibliographystyle{ACM-Reference-Format}
\bibliography{ICCPS/sample_base}

\appendix
\pagebreak
\section*{Appendix}

\setcounter{counterexample}{0}

\subsection{Counterexamples to Monotonic Safety in AEBS}\label{app:cexsAEBS}

In all counterexamples, we set $\timestep = 1s$ and window size $N_{F}=1$. We adopt the notation $Pr(\oursafeprop ~|~ d, s)$ as a shorthand for $Pr_\mdl(d,s)(\oursafeprop)$. 

\begin{counterexample}\label{th:ce1}
Assumption \mosdist does not hold in the AEBS with one BP and distance-dependent perception for shift $\Delta = (1 m, 0)$ when $B = 10 m/s^2, v_0 = 11 m/s, d_0 = 13m,$ and $\detprob(d) = 1 - \lceil d\rceil/20$.
\end{counterexample}
\begin{proof}
We set \oldstate = (13, 11) and \newstate = (14, 11). Then we calculate as follows.
\begin{align*}
    &Pr_{\mdl(\newstate)}(\oursafeprop) = \\
    &  \detprob(14) * Pr(\oursafeprop|3, 1) + (1 - \detprob(14)) *\underbrace{Pr(\oursafeprop|3,11)}_{=0} = \\ 
    & \detprob(14) *( \detprob(3) * \underbrace{Pr(\oursafeprop|2,0)}_{=1} \\
    & \; \; \; \; + (1-\detprob(3))*\underbrace{Pr(\oursafeprop|2,1)}_{=\detprob(2)}) = \\
   & \detprob(14) *(\detprob(3) + (1 - \detprob(3))* \detprob(2)) = \\
   & \detprob(14) *(\detprob(3) + \detprob(2) - \detprob(3)*\detprob(2)). 
\end{align*}
\begin{align*}
    &Pr_{\mdl(\oldstate)}(\oursafeprop) = \\
    &  \detprob(13) * Pr(safe|2, 1) + (1 - \detprob(13)) *\underbrace{Pr(safe|2,11)}_{=0} = \\ 
    & \detprob(13) *( \detprob(2) * \underbrace{Pr(safe|1,0)}_{=1} + \\
    & \; \; \; \; (1-\detprob(2))*\underbrace{Pr(safe|1,1)}_{=0}) = \\
  & \detprob(13) * \detprob(2). 
\end{align*}

If we set $\detprob(d) = 1 - \lceil d\rceil/20$ (i.e., the detection probability grows linearly from 0\% at 20 meters to 100\% at 0 meters), then the above falsifies MoS: $$Pr_{\mdl(\oldstate)}(\oursafeprop) = 0.315 > 0.2955 = Pr_{\mdl(\newstate)}(\oursafeprop)$$

P.S. This inequality is reversed for less steep detection curves like $1 - \lceil d\rceil/40$. 
\end{proof}

The second counterexample shows that going at a faster speed may be beneficial because it leads to closer distances, which result in stronger braking.

\begin{counterexample}\label{th:ce2}
Assumption \mosspeed does not hold in the AEBS  with multiple BPs and distance-independent perception for shift $\Delta = (0, -1m/s)$ when $B(d) = [ 10 m/s^2\text{ if } d\leq 11m; 3 m/s^2\text{ otherwise}], v_0 = 9 m/s, d_0 = 20m,$ and $\detprob = 0.5$.
\end{counterexample}
\begin{proof}

We set \oldstate = (20, 9) and \newstate = (20, 8). Then we calculate as follows:
\begin{align*}
    &Pr_{\mdl(\newstate)}(\oursafeprop) = \\
    & \detprob * Pr(\oursafeprop|11, 6) + (1 - \detprob) * Pr(\oursafeprop|11,9)= \\
    & \detprob * (\detprob * \underbrace{Pr(\oursafeprop|5,0)}_{=1} + (1-\detprob) *\underbrace{Pr(\oursafeprop|5,6)}_{=0}) + \\
    & \quad (1-\detprob)*(\detprob*\underbrace{Pr(\oursafeprop|2,0)}_{=1} \\
    & \quad \quad + (1-\detprob)* \underbrace{Pr(\oursafeprop|2,9)}_{=0} = \\
    & (\detprob)^2 + (1-\detprob)*\detprob = \detprob
\end{align*}
\begin{align*}
    & Pr_{\mdl(\oldstate)}(\oursafeprop) =  \\
    & \detprob * Pr(\oursafeprop|12, 5) + (1 - \detprob) * Pr(\oursafeprop|12,8)= \\
    & \detprob * (\detprob * Pr(\oursafeprop|7,2) + (1-\detprob) *\underbrace{Pr(\oursafeprop|7,5)}_{=\detprob}) + \\
    & \quad (1-\detprob)*(\detprob*\underbrace{Pr(\oursafeprop|4,5)}_{=0} \\
    & \quad \quad + (1-\detprob)* \underbrace{Pr(\oursafeprop|4,8)}_{=0} = \dots = \\
    & \detprob^2 (1 + 2 \detprob - 3 \detprob^2 + \detprob^3)
\end{align*}

It can be shown that for any $0<\detprob<1$, $\detprob > \detprob^2 (1 + 2 \detprob - 3 \detprob^2 + \detprob^3)$. For example, for $\detprob = 0.5$, $0.5 > 0.34375.$ Thus, the above falsifies MoS: $$Pr_{\mdl(\oldstate)}(\oursafeprop) >Pr_{\mdl(\newstate)}(\oursafeprop)$$

\end{proof}

\subsection{Counterexamples to Monotonic Safety in Water Tanks}\label{app:cexsWT}

\begin{counterexample}\label{th:ce-tank}
Assumption \moswl does not hold for a water tank system with $\tankcount = 1, \tanksize = 100, \outflow=3, \inflow=40, \timebound = 4, Pr(\wlp = 0) = 0.4,$ and $Pr(\wlp = \tanksize) = 0.6, \wl=10,$ and $\Delta_\wl = 30$. 
\end{counterexample}
\begin{proof}
For state $\oldstate = 10$, we compute the chance of underflow as $0.6^4 = 0.1296$, and the chance of overflow is the chance of 3 or 4 successes in Binomial distribution $Bi(4, 0.4)$, which is $0.1792$. So $Pr_{\mdl(\oldstate)}(\tanksafeprop) = 0.6912$. For state $\newstate =40$, the chance of underflow is 0. The chance of overflow is the chance of 2+ successes in Binomial distribution from the binomial distribution $Bi(4, 0.4)$, so $Pr_{\mdl(\newstate)}(\tanksafeprop) = 0.4752$.  Thus, $\newstate \moswl \oldstate$, but $Pr_{\mdl(\newstate)}(\tanksafeprop) < Pr_{\mdl(\oldstate)}(\tanksafeprop)$.
\end{proof}

\subsection{Proof of Trace Decomposition Lemma}\label{lem:traceDecompProof}

\begin{lemma}[Satisfaction Probability by Trace Decomposition]\label{lemma:traceDecomp}
Let $\mdl$ be a PA. Assume that there are no infinite paths through \mdl. Let $\sigma$ be a scheduler of $\mdl$. Then
\begin{align*}
    Pr_{\mdl}^{\sigma}(\psi) = Pr_{\mdl}^{\sigma}(\psi \& \neg \diamond s) + Pr_{\mdl}^{\sigma}(\diamond s) Pr_{\mdl(s)}^{\sigma}(\psi)
\end{align*}
\end{lemma}
\begin{proof}
We can split up the traces through $\mdl$ by taking the ones that don't pass through state $s$ and the ones that do, as these are two disjoint sets. So
\begin{align}\label{eq:sDecomp}
    Pr_{\mdl}^{\sigma}(\psi) = Pr_{\mdl}^{\sigma}(\psi \& \neg \diamond s) + Pr_{\mdl}^{\sigma}(\psi \& \diamond s)
\end{align}
Now we examine the $Pr_{\mdl}^{\sigma}(\psi \& \diamond s)$ term. Note that $\psi$ is a safety property, so it can be expressed as $\square \neg \rho$. First, we decompose this term using the chain rule for probabilities:
\begin{align}\label{eq:phiAndsDecomp}
    Pr_{\mdl}^{\sigma}(\psi \& \diamond s) = Pr_{\mdl}^\sigma(\diamond s) Pr_{\mdl}^{\sigma}(\psi | \diamond s)
\end{align}

We now focus on the $Pr_{\mdl}^{\sigma}(\psi | \diamond s)$. This asks for the measure of the traces which satisfy $\psi$ among the traces which pass through $s$. But by the construction of $\mdl$ every state which falsifies $\psi$ (i.e. every state for which $\rho$ is true) is a sink state, so it is impossible for a trace through $\mdl$ to falsify $\psi$ and then pass through $s$. Thus we have that:
\begin{align} \label{eq:phiGivensDecomp}
    Pr_{\mdl}^{\sigma}(\psi | \diamond s) = Pr_{\mdl(s)}^{\sigma} (\psi)
\end{align}
    
Substitute \Cref{eq:phiAndsDecomp} and \Cref{eq:phiGivensDecomp} into \eqref{eq:sDecomp} to complete the proof.

\end{proof}

\end{document}